\documentclass[lettersize,journal]{IEEEtran}

\usepackage{array}
\usepackage{textcomp}
\usepackage{stfloats}
\usepackage{url}
\usepackage{verbatim}
\usepackage{cite}
\usepackage{csquotes}
\usepackage{url}
\usepackage{amsmath,amssymb,amsfonts}
\usepackage{comment}
\usepackage{color}
\usepackage{graphicx}
\usepackage{epsfig} 
\usepackage{textcomp}
\usepackage{xcolor}
\usepackage{flafter}
\usepackage{sidecap}
\usepackage{comment}
\usepackage{caption}
\usepackage{mathtools}
\usepackage{tikz}
\usepackage{psfrag}
\usepackage{float}
\usepackage{epstopdf}
\usepackage{subcaption}
\usepackage{tabularx,booktabs}
\usepackage{multirow}
\usepackage{amsthm}
\usepackage{nomencl}
\usepackage{stackengine}
\usepackage{amssymb}
\usepackage{dsfont}
\usepackage[export]{adjustbox}
\makenomenclature
\newtheorem{theorem}{Theorem}

\newtheorem{lemma}[theorem]{Lemma}
\theoremstyle{definition}

\theoremstyle{remark}
\newtheorem{remark}{Remark}

\newtheorem*{problem*}{{Problem Statement}}
\newcommand\norm[1]{\left\lVert#1\right\rVert}
\usepackage{enumitem} 

\begin{document}

\title{An Energy-Efficient Lyapunov-Based Cooperative Adaptive Cruise Controller for Electric Vehicles}

\author{ Hamed Faghihian$^{1}$, Parisa Ansari Bonab$^{1}$, Arman Sargolzaei$^{1}$
\thanks{$^{1}$ Hamed Faghihian, Parisa Ansari Bonab, and Arman Sargolzaei are with the Department of Mechanical and Aerospace Engineering, University of South Florida, Tampa, FL 33620, USA. Emails: hfaghihian@usf.edu}%
}
\date{April 2026}

\maketitle

\baselineskip 11.53pt

\begin{abstract}
As electric vehicles (EVs) are increasingly adopted as platforms for connected and automated vehicles (CAVs), enhancing their energy efficiency becomes critical. With the emergence of vehicle-to-vehicle (V2V) communication, cooperative adaptive cruise control (CACC) offers improved traffic flow, safety, and energy efficiency by enabling real-time coordination among EVs. However, conventional CACC algorithms neglected acceleration and regenerative braking dynamics in their implementation. To address this gap, this paper proposes a third-order dynamic model for EVs which has been derived from real-world experimental data. We also propose a novel, practical, and energy-efficient Lyapunov-based CACC controller explicitly designed for EV platoons. The proposed controller is requiring lower control gains while ensuring string stability and energy efficiency. To validate its effectiveness, we conduct both simulation and experimental environments, demonstrating that our approach reduces velocity fluctuations, maintains string stability at lower headway times, and improves energy efficiency of the CACC platoon by up to 38.5\% compared to a baseline CACC.
\end{abstract}

\providecommand{\keywords}[1]
{
  \small	
  \textbf{\textit{Keywords---}} #1
}

\section{Introduction} \label{ssec:intro}

In recent years, advancements in sensors, communication technologies, and high-precision positioning technologies push the automotive industries to a significant transformation toward autonomous driving. Connected and autonomous vehicles (CAVs) are equipped with advanced sensors, perception systems, and artificial intelligence-based control algorithms, enabling self-driving capabilities \cite{shen2024sequential}. These vehicles utilize vehicle-to-vehicle (V2V) communication technology to communicate with other vehicles, enhancing coordination, safety, and overall efficiency in modern transportation networks \cite{nourinejad2023optimal}. Since EVs are already equipped with most of the necessary technologies for CAVs, many CAVs are now being built on EV platforms \cite{zhang2021novel}.

Although EVs produce zero emissions, they still face range anxiety issues \cite{zhang2023energy}. Increasing battery capacity to address range issues presents a trade-off, as it raises concerns about the environmental impact and emissions, and higher energy consumption in associated with larger batteries \cite{hung2021regionalized}. To mitigate the impact of limited driving range due to battery constraints, researchers explored various solutions, such as charging stations placement, driving and braking torque distribution and control \cite{FAGHIHIAN20241}, route planning, driving behaviors optimization, and driving motors’ efficiency optimization \cite{zhang2019energy}. While EVs are widely regarded as highly energy-efficient, there remain significant opportunities to further enhance the energy efficiency of electric CAVs (E-CAVs) \cite{faghihian2023energy}. Additionally, as the penetration of E-CAVs increases, the demand for transportation rises due to the appeal of zero-emission vehicles with autonomous features. This surge in demand further highlights the range limitations of E-CAVs \cite{faghihian2024introduction}. Therefore, these facts highlight the urgent need to improve the energy efficiency of E-CAVs.

To increase energy efficiency of E-CAVs, several recent studies conducted. For instance in \cite{zhang2021energy} an energy-saving optimization and control (ESOC) framework proposed, which integrates multiconstraint driving scene data with powertrain control by using traditional model predictive control (MPC) and dynamically adjusting powertrain operations. Also in \cite{zhang2021eco} authors developed a three-power nonlinear controller within their eco-cruise control (ECC) system, leveraging real-time slope estimations to adjust motor torque. Their system demonstrated significant energy savings compared with conventional cruise control. However these works showing promising opportunity, but none of them studied vehicle-following scenarios.
 
 In most studies on vehicle-following scenarios, energy efficiency has not been a priority in many approaches. This addressed by proposing a two-stage decision-making model integrating safety zones and powertrain optimization \cite{zhang2021safe}. Their results demonstrated improved energy efficiency without compromising safety, marking a significant advancement over traditional vehicle-following strategies that prioritize one objective over the other. Despite successful demonstrations of E-CAV capabilities in improving energy efficiency and safety across various scenarios and platooning, most studies have not specifically focused on cooperative adaptive cruise control (CACC). Given that a significant portion of automated driving occurs in this mode, optimizing CACC can have a substantial impact on E-CAV energy efficiency.

The enhancement of adaptive cruise control (ACC) through the use of V2V communication leads CACC, which is a crucial component of CAVs \cite{ansari2024secure}. There is a large number of studies in literature in which authors tried to increase energy efficiency or decrease fuel consumption of CACC through different control method or under different circumstances \cite{chen2021economic,zhang2020cooperative,doi:10.1177/0361198120918572}. The study in \cite{doi:10.1177/0361198120918572} demonstrates that CACC can reduce fuel consumption by controlling acceleration and deceleration variations. The effectiveness of CACC in a mixed scenario with vehicles using ACC and CACC presented in \cite{shen2023energy} in which authors demonstrated that connectivity-enabled controllers can achieve up to 30\% energy savings. In \cite{shao2017robust}, a real-time, practical, and robust eco-CACC system was introduced, achieving up to a 23.5\% improvement in fuel efficiency, highlighting its potential for enhancing energy efficiency. However, these studies do not account for string stability and safety, which are critical for practical deployment.

In \cite{besselink2017string}, a delay-based spacing policy within a CACC employed to maintain string stability under various disturbances to minimize energy losses and enhancing fuel efficiency across the vehicle platoon using proportional-derivative (PD) gains. 
A real-time predictive cruise control (ED-PCC) system for eco-driving was developed, which utilizing Pontryagin's Minimum Principle (PMP) to optimize engine torque and gearshift to decrease fuel consumption reported in \cite{chen2018real}. The proposed method achieved up to 8\% fuel savings by utilizing smoother acceleration profiles and reducing unnecessary braking in car-following scenarios. Although the reported fuel savings were modest, a separate study in \cite{liu2021freeway} analyzing CACC in mixed traffic demonstrated that CACC could reduce fuel consumption by up to 50\% compared to ACC at full market penetration. However, these studies lack checking the effect of headway time on energy efficiency and string stability. In \cite{bian2019reducing}, a multiple-predecessor following (MPF) strategy in a CACC framework proposed, enhancing energy efficiency by maintaining string stability through a constant time headway policy. These studies revealing the effect of headway time and string stability on increasing energy efficiency; however, still they lack showing how string stability will affect energy efficiency. CACC capability in decreasing energy consumption is not limited to passenger vehicles. The study in \cite{smith2020analysis} details the development and implementation of a CACC system on heavy-duty trucks in which string stability demonstrated as a vital aspect that ensures a string of vehicles improves distance attenuation rejection, energy efficiency, and safety measures. Despite these studies discussing how CACC can contribute both in safety and energy efficiency, none of these works targeted E-CAVs and EVs.

Integrating energy efficiency management and safety into CACC for E-CAVs has been reported in some studies. Authors in \cite{chen2020battery} proposed an Eco-CACC system designed for EVs operating near signalized intersections. By optimizing vehicle trajectories, their model significantly reduced energy consumption and travel delay, demonstrating up to 9.3\% energy savings, they extended their works in \cite{chen2024developing} within a smart city framework, achieving up to 23.8\% energy savings. While these works showed the effectiveness of eco-CACC; however, they did not consider platooning.
To improve energy efficiency in vehicle platoons, the work reported in \cite{ma2019predictive} introduced a predictive energy-saving framework using nonlinear model predictive control (NMPC) for CACC platoon. Similarly, \cite{ma2020cooperative} proposed a CACC optimization strategy based on simulated annealing particle swarm optimization (SA-PSO) combined with MPC, enhancing both platoon stability and minimizing energy use. However, most of these studies overlook the switching dynamics between driving and braking in EVs. Unlike internal combustion engine vehicles (ICEVs), EVs exhibit distinct acceleration and deceleration behaviors due to regenerative braking systems (RBS). While \cite{li2016multiple} introduced a switching mode controller for EVs, it relies on predefined controllers rather than a continuous control approach and has not been integrated into CACC. The integration of a truly continuous control strategy into CACC design remains unexplored.
\\

 The absence of a practical EV model that accurately represents both motoring and RBS phases, coupled with the lack of a Lyapunov-based CACC approach, highlights a significant research gap. Moreover, the effect of string stability criteria on energy efficiency has not been systematically explored in the literature.
Unlike previous studies, we propose a novel Lyapunov-based CACC design that ensures safety, string stability, and energy efficiency for EVs. The performance, safety, and string stability of our proposed controller are validated in both simulation and practical environments. Additionally, the results demonstrate a significant reduction in total energy consumption compared to existing controllers.

Thus, the contributions of this paper are
\begin{enumerate}[label=\roman*)]
    \item We Proposed a realistic third-order switched dynamic system differential equation model of an EV (Mustang Mach-E) through a real-world experimental analysis that explicitly models EV acceleration and deceleration separately.
    \item We designed an energy-efficient and novel Lyapunov-based CACC for EVs based on the proposed switched dynamic system model, which is implementable in the real world and requires lower controller gains.The proposed CACC guarantees safety and time-domain string stability at a headway time which is smaller than typically reported for string-stable CACC platoons which implies higher achievable traffic capacity and improved energy efficiency at the system level.
    
    \item  We utilized the Filippov framework for Lyapunov stability analysis and time-domain string stability as a performance measure to evaluate how the proposed CACC enhances both string stability and energy efficiency in EV platoons, and demonstrating its energy efficiency improvement through an experimental comparison with a baseline controller.
\end{enumerate}

The paper is organized as follows: Section \ref{cacc_model} presents the mathematical model of a CACC. Section \ref{problem} outlines the problem statement. The proposed solution is discussed in section \ref{S.proposed}.  Section \ref{s:results} presents the outcomes of the proposed CACC, beginning with the experimental results to identify the EV model. This is followed by simulation and experimental results of the proposed CACC to assess its performance in maintaining a safe and stable distance. The section then analyzes time-domain string stability criteria using a rigorous standard drive cycle and concludes with an experimental energy analysis that compares the energy efficiency of the proposed controller with a baseline CACC. Section \ref{s:conclusion} provides the conclusion and possible future works. Also, the stability analysis of the proposed controller presented in the Appendix section.

\begin{figure*}[] 
\includegraphics[width=\textwidth,height=3.4cm]{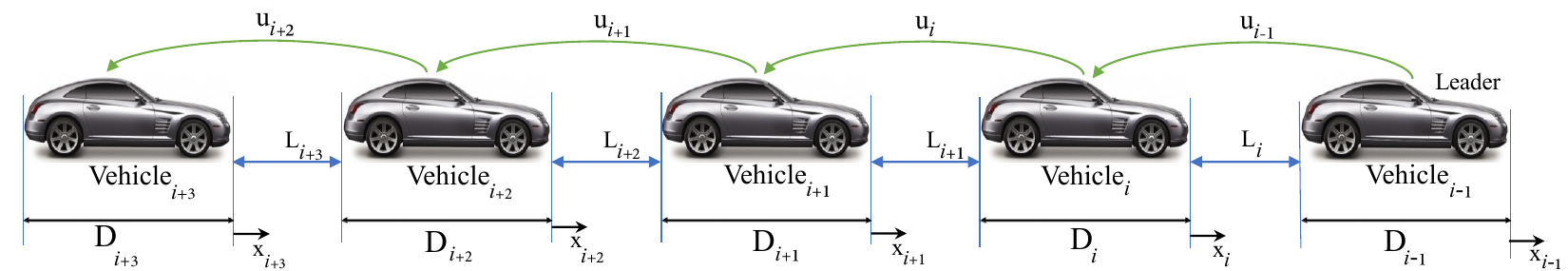}
\caption{Schematic of a string of CACC-equipped vehicles.}
\label{fig:cacc}
\end{figure*}

\section{Mathematical Model of CACC} \label{cacc_model}
Schematic of a string of CACC-equipped in a platoon of vehicles is demonstrated in Figure \ref{fig:cacc}.
For simplicity of analysis, it is assumed that vehicles are homogeneous. 
The dynamic model of $i^{th}$ vehicle ($i\in \{1,\cdots,n\}$, with $n \in \mathbb{N})$ is the number of vehicles) is obtained from experimental analysis as 
  \begin{equation}\label{system_dynamics}
\begin{cases}
\dot{x}_{i}(t) = v_{i}(t),\\
\dot{v}_{i}(t) = a_{i}(t),\\
\dot{a}_{i}(t) = -\gamma_{i}(u_i)a_{i}(t) + \beta_{i}(u_i) u_i(t),
\end{cases}
\end{equation}
where $\gamma_{i}(u_i) \in \mathbb{R}$ and $\beta_{i}(u_i) \in \mathbb{R}$ are defined as
\[
\gamma_{i}(u_i) = \frac{1 + \textit{sgn}(u_i)}{2} \gamma_{i_1} + \frac{1 - \textit{sgn}(u_i)}{2} \gamma_{i_2},
\]
\[
\beta_{i}(u_i) = \frac{1 + \textit{sgn}(u_i)}{2} \beta_{i_1} + \frac{1 - \textit{sgn}(u_i)}{2} \beta_{i_2},
\] where, $x_i\in\mathbb{R}$, $v_i\in\mathbb{R}$, $a_i\in\mathbb{R}$, and $u_i\in\mathbb{R}$ represent the position, velocity, acceleration, and control input, respectively. Also, $\gamma_{i_1}, \gamma_{i_2}\in\mathbb{R}$ and $\beta_{i_1}, \beta_{i_2}\in\mathbb{R}$ are experimentally identified positive constants (see Section \ref{sec:model}) that distinguish acceleration and deceleration dynamics in EVs.

 This distinction in acceleration and declaration behavior in EVs is crucial for accurately modeling EV acceleration and RBS effects, improving the representation of EV behavior during acceleration and deceleration.
 
\begin{remark}
\label{ass:regular_input}
For each vehicle $i$, the control input $u_i (t) : [t_0,\infty) \to \mathbb{R}$
generated by the CACC law in Section~\ref{sec:control_design} satisfies

\begin{enumerate}
    \item $u_i(t)$ is piecewise and continuously differentiable with uniformly
    bounded derivative, i.e., there exists $L_i > 0$ such that
    $|\dot u_i(t)| \leq L_i$.

      \item There exists a constant $\delta_i > 0$ such that between any two consecutive zero crossings of $u_i(t)$, the signal reaches to a magnitude of at
    least $\delta_i$. This means, if $t_k < t_{k+1}$ are two consecutive times with $u_i(t_k) = u_i(t_{k+1}) = 0$ and $\textit{sgn}(u_i)$ changes over these two consecutive times, then there is a time $t^\star \in (t_k,t_{k+1})$ with $|u_i(t^\star)| \geq \delta_i$.
\end{enumerate}
\end{remark} 
Remark~\ref{ass:regular_input} is physically valid since $u_i(t)$
is produced by a continuous CACC controller acting on smooth error and filtering by the drivetrain dynamics, therefore it cannot switch infinitely fast.
Moreover, the traction and regenerative braking controllers always exhibit finite resolution and hysteresis, which imposes a minimum effective torque level associated with $\delta_i$ between two sign changes. In implementation, this corresponds to a small deadband around zero acceleration/torque and prevents high-frequency
dithering near $u_i = 0$.
Because 
$u_i$ is produced by a continuous CACC law and filtered through drivetrain/actuator dynamics, 
it is continuous and bandwidth-limited. Therefore the Remark \ref{ass:regular_input} holds for all scenarios.

\section{Problem Statement} \label{problem}
\subsection{{Safety}}
As shown in Figure \ref{fig:cacc}, vehicles transmit their control signals to followers using wireless communication. This means each vehicle follows the leader ahead of it. The follower vehicle receives data from its lead vehicle through a wireless communication channel, which it uses to match the leader's speed and maintain a safe distance. Thus, the primary objective of this paper's CACC design is to use this communication method to ensure safety and maintain the desired distance between vehicles.
In order to quantify the safety objective, we define distance errors and their first and second derivatives. To define the distance error, first the actual distance between the $i^{th}$ follower and its leader is calculated as 
\begin{equation}\label{Distance_Between_Vehicles}
L_{i}(t) = x_{i-1}(t)-x_{i}(t)-D_{i-1},
\end{equation}
where $L_{i}(t)\in\mathbb{R}$ is the distance between the $i^{th}$ follower and its leader, $D_{i-1}\in\mathbb{R}$ is the length of the lead vehicle. Also, the desired distance between the $i^{th}$ follower and its leader is defined as 
\begin{equation}\label{Desired_Distance_Between_Vehicles}
L_{d_{i}}(t) \overset{\Delta}{=} c_{i}+b_iv_{i}(t),
\end{equation}
where $L_{d_{i}}(t)\in\mathbb{R}$ is the desired distance between vehicles, $c_i \in \mathbb{R}$ is the standstill distance which is a constant, and $b_i \in \mathbb{R}^+$ is the headway time, and $\mathbb{R}^+$  defined as $\mathbb{R}>0$.
Therefore, distance error $e_{i}:[t_0, \infty) \to \mathbb{R}$ is defined as
\begin{equation}\label{Distance_Error}
e_{i}(t)  \overset{\Delta}{=} L_{i}(t)-L_{d_{i}}(t).
\end{equation}
Achieving and maintaining the desired distance is our objective to ensure safety, therefore, as \( t \to \infty \), the distance error \( e_{i}(t) \) must approach zero.

\subsection{String Stability}
Another objective is to ensure the string stability of vehicles in a platoon, focusing on the amplification of oscillations from vehicle $i$ to those that follow. To address this objective and quantify that, the system is defined as string stable if and only if
\begin{equation}\label{ss:1}
||\zeta_i||_m \le ||\zeta_{i-1}||_m,
\end{equation}
where \(\zeta_i \in \mathbb{R}\) can represent the distance error (\(e_i\)), velocity (\(v_i\)), or acceleration (\(a_i\)) being evaluated in the \(i\)-th vehicle, with \(|| . ||_m\) denoting the \(m\)-th norm of the signal.
\subsection{Energy Efficiency}
The CACC has to be designed to reduce total energy consumption. Although there is no direct relationship between any controller parameters and energy efficiency, the designed controller should be verified in terms of energy efficiency. Energy efficiency can be investigated based on the total amount of energy spent through a scenario involving a string of vehicles. In the context of the EV model utilized in this study, the overall energy sourced from the vehicles, \( E_c \in \mathbb{R}\), is described as \cite{10500159}
\begin{equation} \label{equ:eff_1}
E_c = E_{\text{air}} + E_{\text{roll}} + E_{\text{acc}} + E_g - E_{rb},
\end{equation}
where \(E_{\text{air}}\in \mathbb{R}\) represents the energy required to overcome aerodynamic effects, \(E_{\text{roll}}\in \mathbb{R}\) is the energy consumed to overcome rolling resistance, \(E_{\text{acc}}\in \mathbb{R}\) corresponds to the energy expended for vehicle acceleration, \(E_g\in \mathbb{R}\) accounts for the energy consumed in altering the vehicle's elevation, and \(E_{rb}\in \mathbb{R}\) shows the amount of recaptured energy. To minimize the total consumed energy denoted in \eqref{equ:eff_1} in CACC, velocity variations should be minimized \cite{faghihian2023energy}.

\section{Proposed Solution} \label{S.proposed}

\subsection{Error Signals}

To facilitate the stability analysis, the following error signals as the first and second derivative of distance error are defined as
\begin{equation}\label{Errors}
\begin{cases}
e_{i_1}(t) \overset{\Delta}{=} e_{i}(t),\\
e_{i_2}(t)\overset{\Delta}{=}\dot{e}_{i_1}(t),\\
e_{i_3}(t)\overset{\Delta}{=}\dot{e}_{i_2}(t).
\end{cases}
\end{equation}

Using \eqref{Distance_Between_Vehicles} and \eqref{Distance_Error}, $\dot{e}_{i_3}(t)$ can be written as
\begin{equation}\label{Acceleration_error_1}
\dot{e}_{i_3}(t)=\dddot{x}_{i-1}(t)-\dddot{x}_{i}(t)-b_i\dddot{v}_{i}(t),
\end{equation}
using \eqref{system_dynamics} and since $\dddot{v}_i=\ddot{a}_i$, \eqref{Acceleration_error_1} we can have
\begin{equation} \begin{aligned}\label{Acceleration_error_2}
\dot{e}_{i_3}(t)=& -\gamma_{i-1}(u_{i-1})a_{i-1}(t)+\beta_{i-1}(u_{i-1})u_{i-1}(t)\\&+\gamma_{i}(u_{i})a_{i}(t)-\beta_{i}(u_{i})u_{i}(t)\\& -b_i\big(-\gamma_i(u_{i})\dot{a}_i(t)+\beta_i(u_{i})\dot{u}_i(t)\big),
\end{aligned}
\end{equation}
further simplification yields
\begin{equation}\begin{aligned}\label{Acceleration_error_4}
\dot{e}_{i_3}(t)= &\beta_{i-1}(u_{i-1})u_{i-1}(t)\\&-\beta_i(u_{i})(b_i\dot{u}_i(t)+u_i(t)) \\&
-\gamma_{i-1}(u_{i-1})a_{i-1}(t)
+\gamma_{i}(u_{i})a_i(t) \\& 
-\gamma^2_{i}(u_{i})b_ia_i(t)
+b_i\gamma_{i}(u_{i})\beta_i(u_i)u_i(t),
\end{aligned}
\end{equation}
and we can write \eqref{Acceleration_error_4} as
\begin{equation}\begin{aligned}\label{Acceleration_error_5}
\dot{e}_{i_3}(t) &= \beta_{i-1}(u_{i-1})u_{i-1}(t)
\\&-\beta_i(u_{i})(b_i\dot{u}_i(t)+u_i(t)) \\
& -\gamma_{i-1}(u_{i-1})a_{i-1}(t)+\gamma_{i}(u_{i})a_i(t)\\
&+ b_i\gamma_i(-\gamma_ia_i(t)+\beta_i(u_{i})u_i(t)),
\end{aligned}
\end{equation}
considering $a_{i-1}=\ddot{x}_{i-1}$, $a_i=\ddot{x}_i$, and \eqref{system_dynamics}, equation \eqref{Acceleration_error_5} can be written as 
\begin{equation}\begin{aligned}\label{Acceleration_error_6}
\dot{e}_{i_3}(t)= &\beta_{i-1}(u_{i-1})u_{i-1}(t)\\&-\beta_i(u_{i})\big(b_i\dot{u}_i(t)+u_i(t)\big) \\&
-\big[\gamma_{i-1}(u_{i-1})\ddot{x}_{i-1}(t)-\gamma_{i}(u_{i})\ddot{x}_{i}(t)\\&~~~~~~~~~~~~~~~~~~~~~~~-b_i\gamma_{i}(u_{i})(\dot{a}_i(t))\big],
\end{aligned}
\end{equation}

By defining $\phi_i(t) \in \mathbb{R}$ as
\begin{equation}\label{new_equation}
\begin{aligned}
&
\phi_i(t)\triangleq \\&
\big[\gamma_{i-1}(u_{i-1})\ddot{x}_{i-1}(t)-\gamma_{i}(u_{i})\ddot{x}_{i}(t)-b_i\gamma_{i}(u_{i})(\dot{a}_i(t))\big],
\end{aligned}
\end{equation}
therefore, we can write \eqref{Acceleration_error_6} as
\begin{equation} \begin{aligned} \label{its.Acceleration_error_7}
\dot{e}_{i_3}(t)= & \beta_{i-1}(u_{i-1})u_{i-1}(t)-\beta_i(u_{i})(b_i\dot{u}_i(t)+u_i(t)) \\& -\phi_i(t),
\end{aligned}
\end{equation}
we can write \eqref{its.Acceleration_error_7} as
\begin{equation}\label{Acceleration_error_8}
\dot{e}_{i_3}(t)= -\phi_i(t)-P_i(t)+\beta_{i-1}(u_{i-1})u_{i-1}(t).
\end{equation}

where, $P_i(t) \in \mathbb{R}$, defined as
\begin{equation}\label{its.controller.P}
    P_i(t)\triangleq \beta_i(u_{i})\big(b_i\dot{u}_i+u_i\big).
    \end{equation} 

For simplicity of notations, $t$ and $u_i$ is removed in the equations of the following sections.
\vspace{-5pt}
\subsection{Control Design}\label{sec:control_design}
 The filtered control signal, \( P_i(t) \), must stabilize the error dynamics in \eqref{Acceleration_error_8}. Therefore, to derive the control signal \( u_i(t) \) for CACC from \eqref{its.controller.P}, the control law for \( P_i(t) \) was developed based on the Lyapunov stability analysis in Appendix as
\begin{equation}\begin{aligned}\label{control_signal}
P_i= &(\alpha_{i_1}+\alpha_{i_2})e_{i_3}+\beta_iC_ir_{i_2}+\beta_{i-1}u_{i-1} \\& + (\alpha_{i_1}\alpha_{i_2}+1)r_{i_1} - \alpha_{i_2}\alpha_{i_1}^2e_{i_1}-\phi_i,
\end{aligned}
\end{equation}
where $C_i\in\mathbb{R}^+$ is a gain specified by the user, and two auxiliary error signals $r_{i_1}\in\mathbb{R}$ and $r_{i_2}\in\mathbb{R}$ are defined as
\begin{equation}\label{auxiliary_error}
r_{i_1}(t) \overset{\Delta}{=} \dot{e}_{i_1}(t) + \alpha_{i_1}e_{i_1}(t),
\end{equation} 
\begin{equation}\label{auxiliary_error_1}
r_{i_2}(t) \overset{\Delta}{=} \dot{r}_{i_1}(t) + \alpha_{i_2}r_{i_1}(t),
\end{equation}
where $\alpha_{i_1}$and $\alpha_{i_2}\in\mathbb{R}^+$ are defined as user-specified known gains. By selecting appropriate coefficients, the proposed controller guarantees a safe distance between vehicles.
\vspace{-5pt}
\subsection{String Stability Assessment}
 
Given the practical limitations of frequency-domain methods in experimentally evaluating string stability, and the inability to derive complementary sensitivity for switched dynamic system, we propose the adoption of time-domain approaches. This method offers a more feasible framework for experimental analysis, as evidenced by recent studies \cite{alipour2020impact}. Using same principles Considering \eqref{ss:1} the string stability criterion, \(||\Omega_i||_2 \in \mathbb{R}\), can be defined as \begin{equation}\label{ss:our_def}
||\Omega_i||_2\triangleq\frac{||{v}_i(t)||_2}{||v_{i-1}(t)||_2}\leq1, ~v_{i-1}\neq0,~i\geq2.
\end{equation}
The proposed criteria in \eqref{ss:our_def} can be used as an evaluation factor. This criterion used as a reference to evaluate effectiveness of the designed controller in Section \ref{s:results}.

\vspace{-5pt}
\subsection{Energy Efficiency}
 To investigate the effect of the controller on the energy efficiency of an EV, we can compare measurements of the total energy consumed by the EV under the same scenario but with different controllers. The total vehicle energy consumption, \(E_C\), introduced in \eqref{equ:eff_1}, is obtained experimentally by integrating the real-time power over time from the start time, \( t_0 \in \mathbb{R} \), to the end time, \( t_f \in \mathbb{R} \), of the drive cycle, given by
\begin{equation} \label{its.eq.power_integral}
E_C = \int_{t_0}^{t_f} P(t) \, dt,
\end{equation}
where \(P(t) \in \mathbb{R}\), is the Real-time power data obtained from the EV power management system.

\section{Results} \label{s:results}
In the results section, we first demonstrate how we derived the model of the EV. Next, by using the derived model, we present the simulation and experimental results of a test scenario to assess CACC performance and tuning a minimum headway time. Following this, we show results from a rigorous standard drive cycle to evaluate the string stability of the designed controller. Finally, we present the results of a designed test scenario to assess both the energy efficiency and string stability of the proposed controller.
\vspace{-10pt}
 \textbf{ \subsection{\textbf{ Vehicle Model Through Experimental Analysis}}\label{sec:model}}

The EV longitudinal model in \eqref{system_dynamics} identified experimentally using
the Mustang Mach-E experiment platform shown in Figure~\ref{fig:mach-e}. We applied step desired
acceleration/deceleration commands, performed 25 trials in motoring mode and 40 trials
in deceleration mode, scaled and averaged the measured acceleration responses. The resulting mean step responses for $u_i \ge 0$ and $u_i < 0$  shown in Figure~\ref{its.fig.test1}. Fitting a first-order model to each mode yields the values for the model presented in \eqref{system_dynamics} as
$\beta_{i1}= 0.7378$, $\gamma_{i1} = 0.6998$, $\beta_{i2}= 0.9315$, and $\gamma_{i2} = 0.9009$.

\begin{figure}[]  
    \centering
    \includegraphics[width=.45\textwidth]{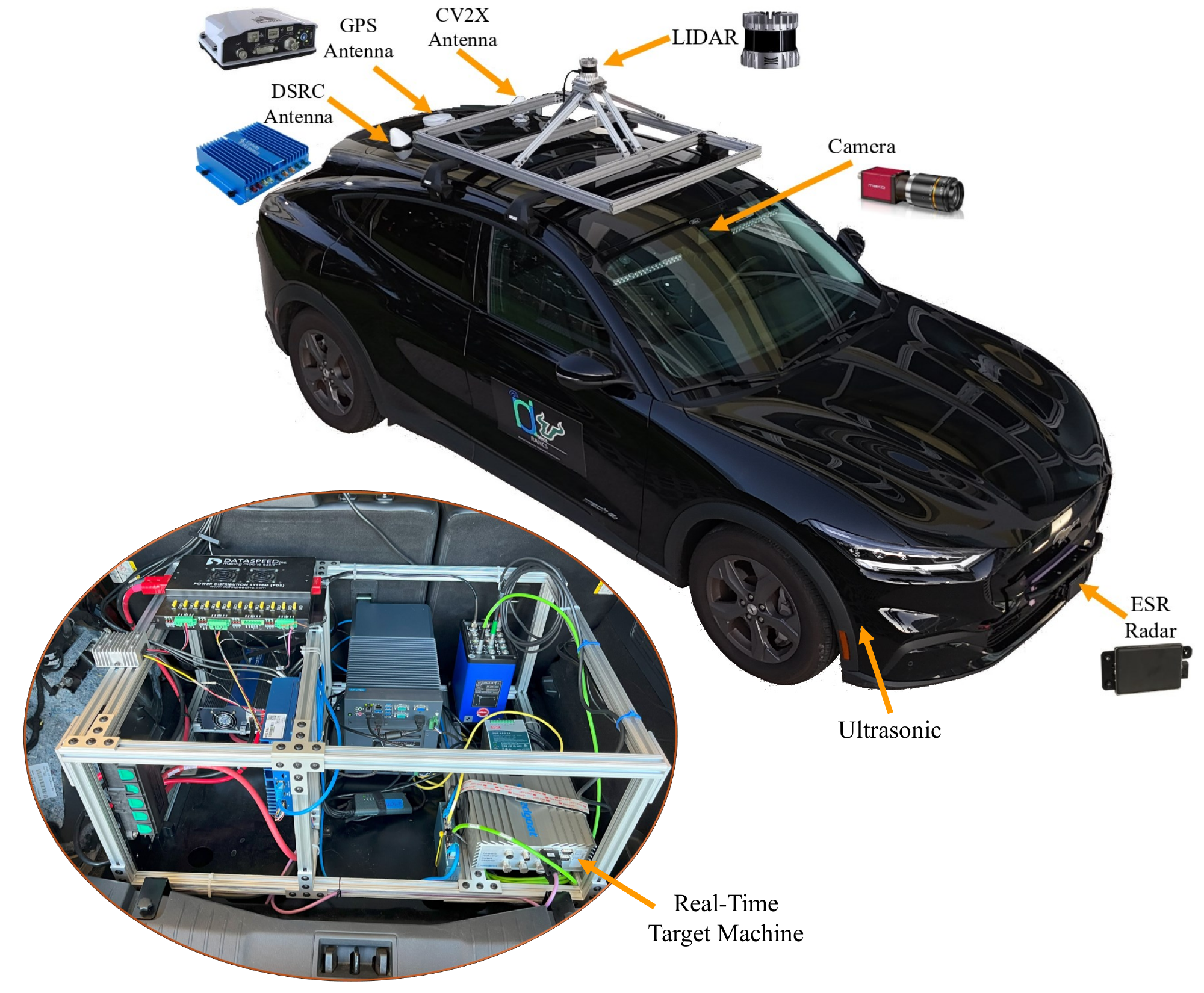}
    \caption{The EV for experimental test equipped with ADAS features, Cellular-vehicle-to-everything (C-V2X) communication, Dedicated short-range communications (DSRC), light detection and ranging (LIDAR), multi-mode electronically scanning radar (ESR), and real-time target machine}.
   \label{fig:mach-e}
\end{figure}

\begin{figure}[] 
     \centering
    \begin{subfigure}{.34\textwidth}
        \includegraphics[width=\linewidth]{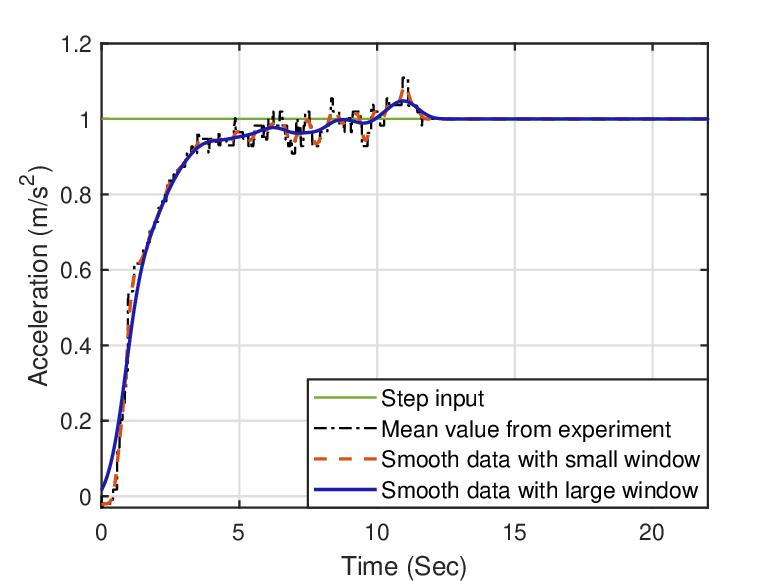}              \caption{}
        \
        \end{subfigure}
    \begin{subfigure}{.34\textwidth}
            \includegraphics[width=\linewidth]{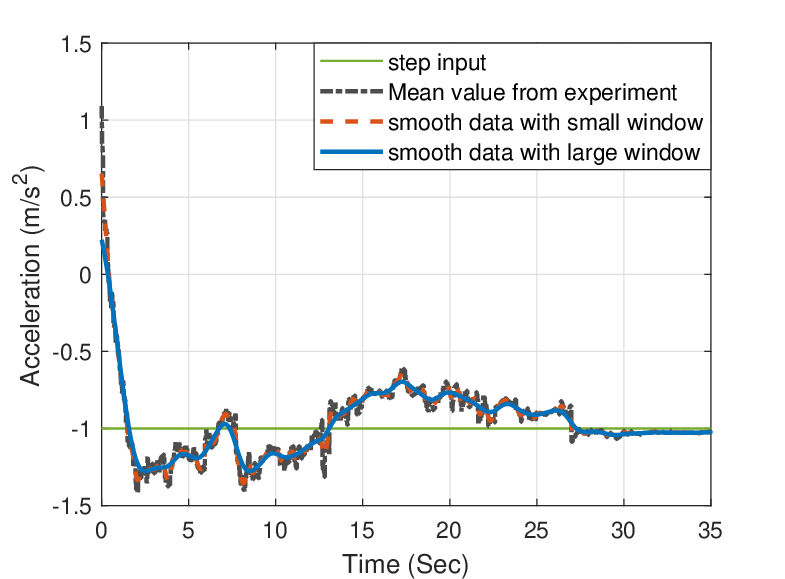}\caption{}\vspace{-5pt}
    \end{subfigure}
      \caption{Obtained experimental acceleration-step response
    (a) In acceleration
    (b) In deceleration}
            \label{its.fig.test1}
        \end{figure}

\subsection{\textbf{CACC Performance Evaluation (Scenario I)}} \label{subsec:sim1}

The controller was tested through simulation and experiment to ensure it handles CACC safety measures and headway time tuning. Lower headway times increase fluctuations, while higher times create undesirable gaps between vehicles. Identifying the minimum headway time that maintains safety with minimal fluctuations improves both energy and road efficiency. To determine the minimum headway time for CACC performance with reasonable error, we began testing at 0.1 seconds and gradually increased it until negligible fluctuations were observed. The desired acceleration for the lead vehicle for this test scenario is shown in Figure \ref{fig:Desired_Acceleration}.

\begin{figure}[]
    \centering
    \includegraphics[width=.35\textwidth]{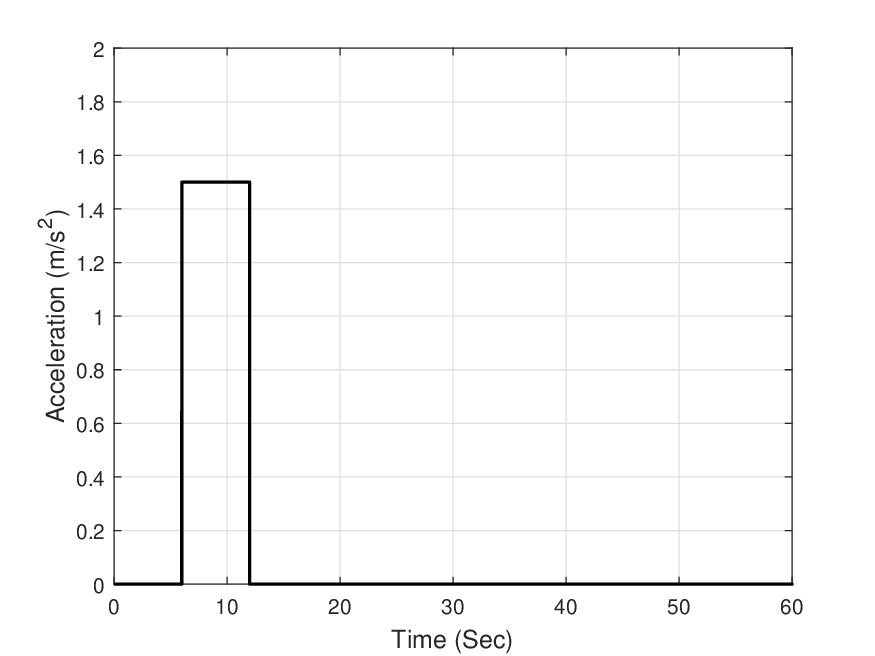}
    \caption{Desired acceleration input for the lead vehicle in the scenario I.}
    \label{fig:Desired_Acceleration}
    \end{figure}
    
\subsubsection{Simulation results for scenario I}

 The speed profile of all vehicles with $0.2$ seconds headway time is shown in Figure \ref{fig:speed_all_1}. As shown, vehicles cannot maintain their leaders' speeds, resulting in fluctuations. These fluctuations are undesirable as they increase the risk of crashes and reduce energy efficiency. Testing various headway times revealed that 0.5 seconds is the minimum headway time to ensure all followers move with negligible fluctuation. The speed of all vehicles with 0.5 seconds of headway time is illustrated in Figure \ref{fig:speed_all}. As shown, even the last vehicles reached the same value for the velocity in a reasonable time, without fluctuations. Position errors of two first and two last vehicles are depicted in Figure \ref{fig:errorFL}.
 \begin{figure}[h]
    \centering
    \includegraphics[width=.35\textwidth]{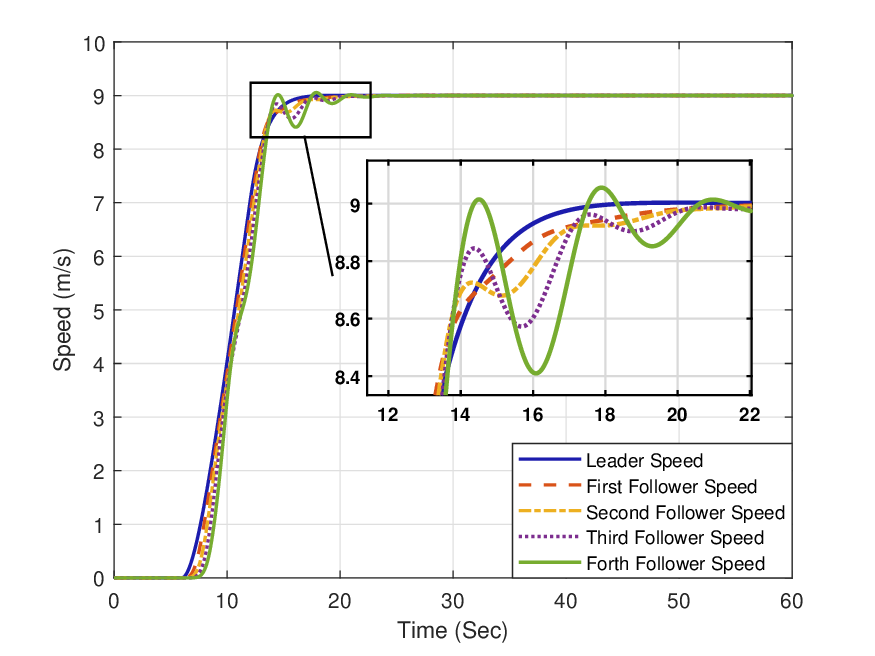}
    \caption{Speed profile of all vehicles with 0.2 seconds headway time in scenario I.}
     \label{fig:speed_all_1}
\end{figure} 
\begin{figure}[]
    \centering
\includegraphics[width=.35\textwidth]{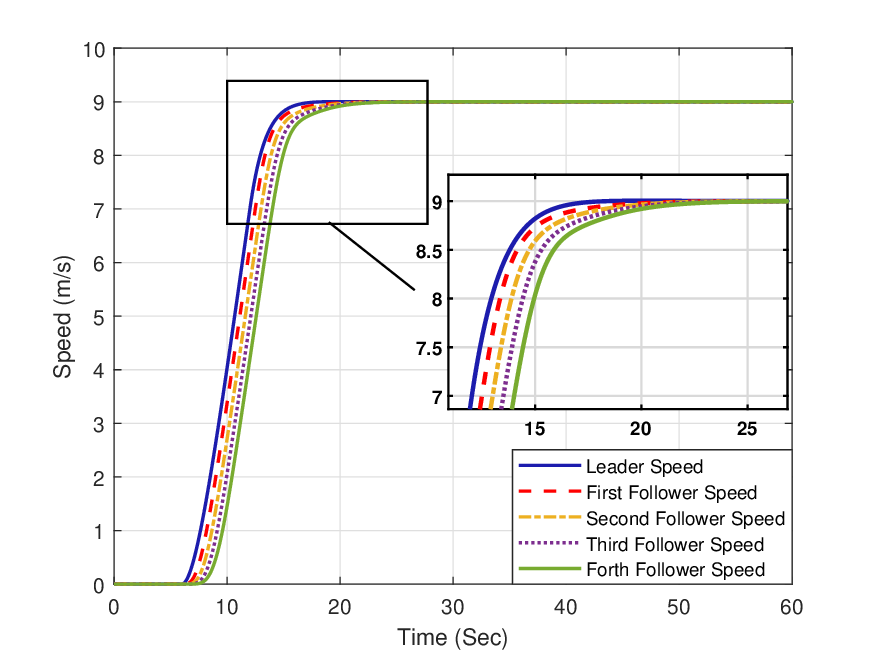}
    \caption{Speed profile of all vehicles with 0.5 seconds headway time in scenario I}.
      \label{fig:speed_all}
\end{figure}

\begin{figure}[] 
     \centering
    \begin{subfigure}{.241\textwidth}
        \includegraphics[width=\linewidth, right]{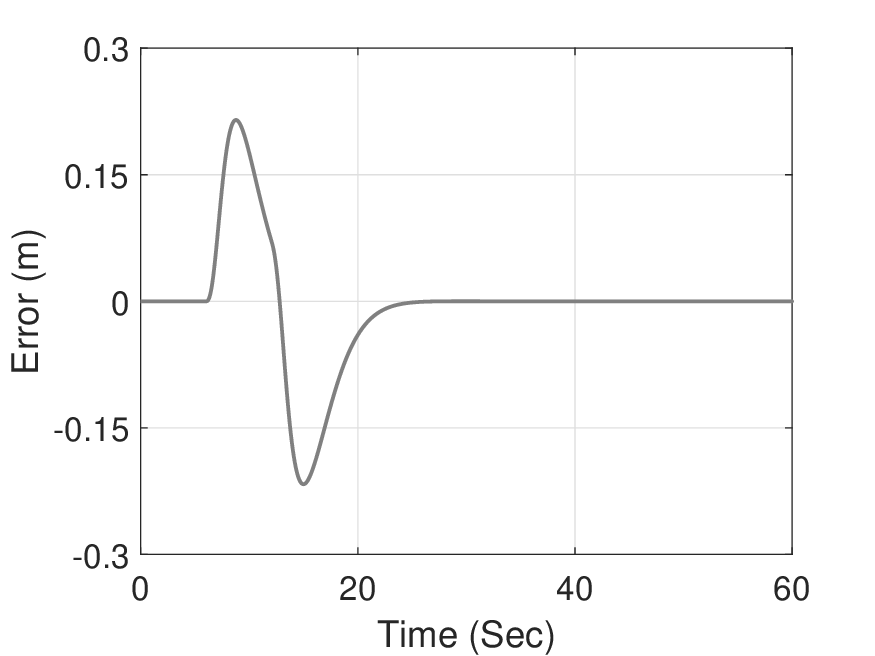}              \caption{}
    \end{subfigure}
    \begin{subfigure}{.241\textwidth}
            \includegraphics[width=\linewidth, left]{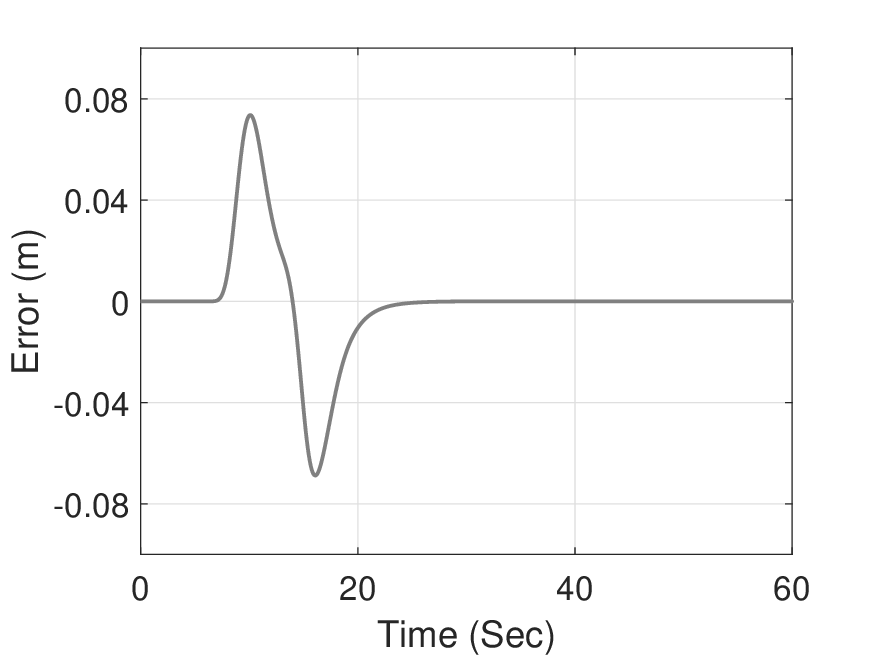}\caption{}
    \end{subfigure}
\caption{(a)  Position Error of two first vehicles 
    (b) Position Error of the two last vehicles. } 
    
        \label{fig:errorFL}
\end{figure}
Based on this figure, the position error quickly stabilizes to zero after a brief fluctuation, indicating that the designed controller can effectively compensate for changes in vehicle velocities and acceleration. The distance error maximum value decreases along the string. The fact that the distance error between the last two vehicles is less than the length of the vehicle ensures that it is possible to control the whole string with this controller.
 
\subsubsection{Experimental results for scenario I}
\

Here, we present the results of implementing the designed controller on the vehicle shown in Figure \ref{fig:mach-e}. We used the vehicle-in-the-loop (VIL) as the second vehicle, as shown in Figure \ref{fig:experiment_sc1}. 
The desired acceleration for the leader is the same as that used in the simulation, shown in Figure \ref{fig:Desired_Acceleration}.
The designed controller and simulated environment are implemented on a real-time machine, enabling realistic online calculations based on live data from vehicle sensors. This setup constitutes an experimental VIL test.. A comparison of the speed of the first follower with the lead vehicle is depicted in Figure \ref{fig:experiment}.

As shown, the VIL experiment shows that the proposed controller is implementable on the EV using vehicle onboard sensors. The comparison of the average root mean square error (RMSE) of the string of five vehicles is shown in Table \ref{cacc.table1}.
\begin{figure}[h]
\vspace{-5pt}
    \centering
\includegraphics[width=.31\textwidth]{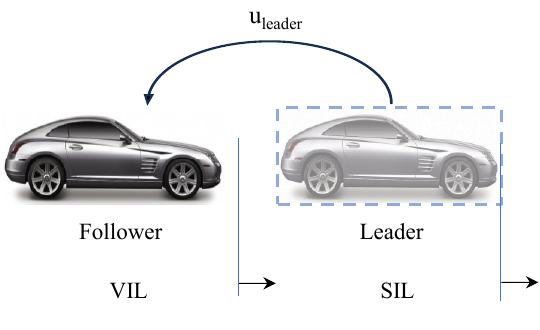}
\caption{Configuration of vehicles for experimental testing in scenario I.}
    \label{fig:experiment_sc1}
\end{figure}
\begin{figure}[h]
\vspace{-5pt}
    \centering
    \includegraphics[width=.36\textwidth]{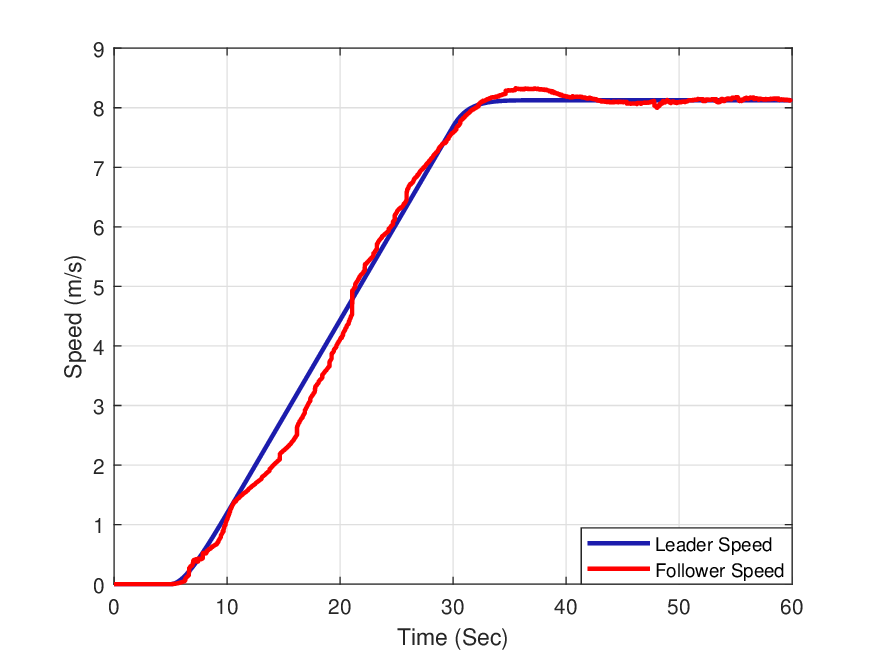}
    \caption{Comparison of the speeds of the leader and the follower in the experimental testing in scenario I. }
    \label{fig:experiment}
\end{figure}
\begin{table}[h!] 
\centering
\caption{Comparison of average of position RMSE in scenario I.}\label{cacc.table1}
\begin{tabular}{c c}
\textbf{Configuration and headway time } & \textbf{Distance RMSE} \\ \hline

\\

Simulation, ~\(b=0.5\) & ~~~0.0696                              \\
Experiment, ~\(b=0.5\)   & ~~~0.9486
\end{tabular}
\end{table}

\subsection{\textbf{String Stability Evaluation by Simulation}}

A standard drive cycle test was employed to assess the string stability of the designed controller. The \(US06\) drive cycle, known for modeling aggressive and rigorous driving, was used to evaluate both string stability and the proposed controller's performance. The leader's desired acceleration was set to match the \(US06\) speed, based on data from \cite{CiteDrive2022}. To analyze the effect of this drive cycle on our proposed controller, the speeds of all vehicles in the CACC with a 0.5-second headway time are shown in Figure \ref{fig:speeds_results}.

\begin{figure}[]
    \centering
    \includegraphics[width=.36\textwidth]{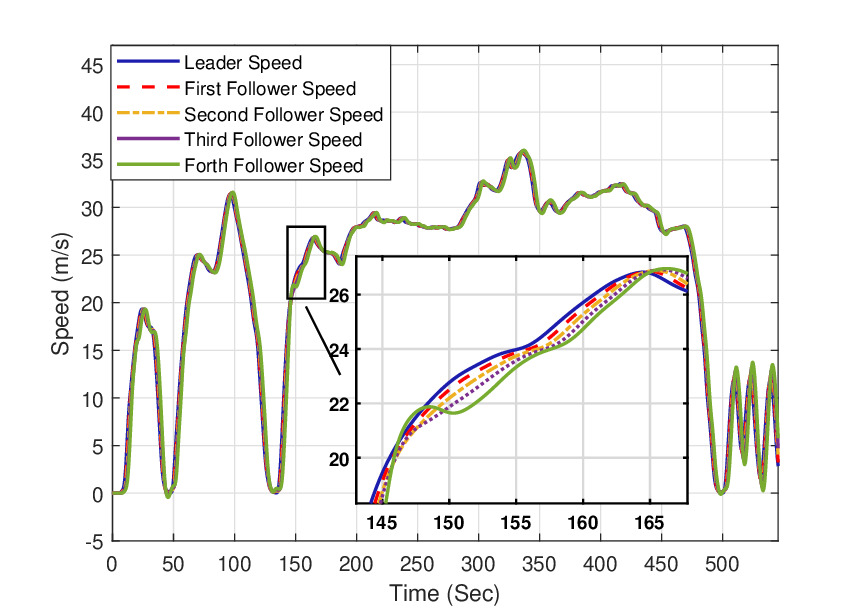}
    \caption{Speed profile of all vehicles under \(US06\) drive cycle test. }
    \label{fig:speeds_results}
    \vspace{-15pt}
\end{figure}
\begin{figure}[h!]
\vspace{-5pt}
    \centering
    \includegraphics[width=.36\textwidth]{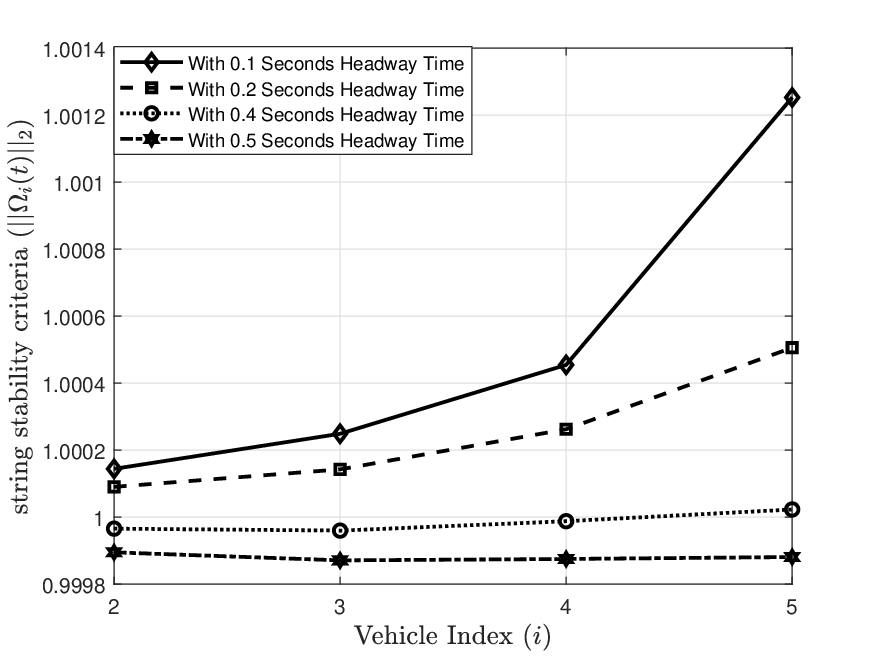}
    \caption{Time-domain string stability for the proposed controller with different headway times, \(b\), using the \(US06\) drive cycle. }
    
    \label{fig:string_stability_analysis}
        \vspace{-15pt}
\end{figure}
Results of the time domain string stability criterion for all of the followers in the CACC with different headway times based on \cite{alipour2017string} are shown in Figure \ref{fig:string_stability_analysis}.

Based on equation \eqref{ss:our_def} the string stability criteria, \(||\Omega_i||_2\), should be less than one and should not increase down the string. As shown, with a 0.2-second headway time, this criterion increases among the string, indicating it cannot guarantee string stability. However, with higher headway times, the string stability criterion does not increase with the number of vehicles, especially for headway times greater than 0.4 seconds, ensuring string stability. At a headway time of 0.5 seconds, the string stability criterion remains within the acceptable limit.
\begin{figure}[]
    \centering
    \includegraphics[width=.36\textwidth]{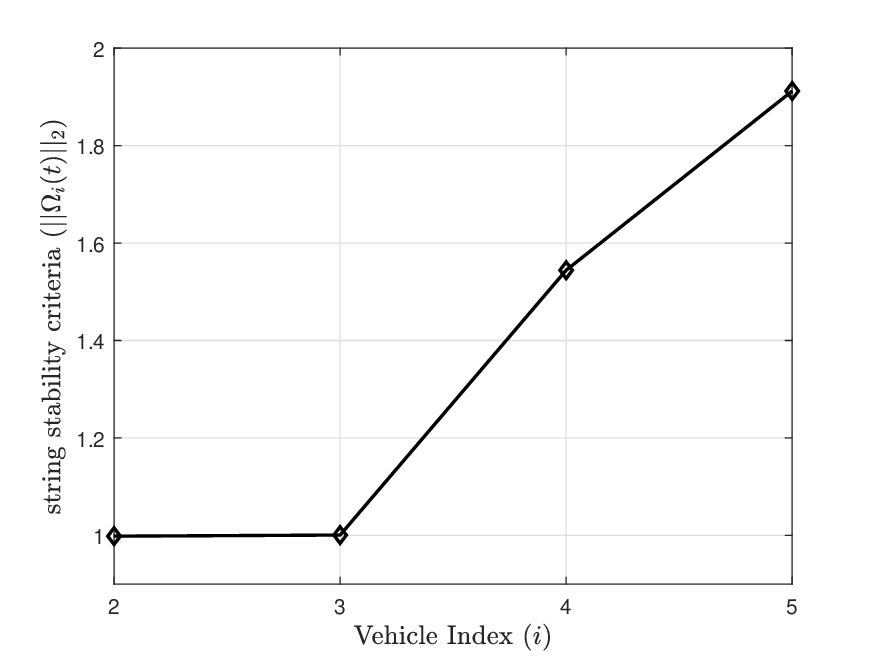}
    \caption{Time-domain string stability with the baseline controller using the \(US06\) drive cycle.}
    \label{fig:stringstabilityanalysispid}
\end{figure}
To compare, the string stability criteria using a baseline PID-based controller with this drive cycle are investigated and shown in Figure \ref{fig:stringstabilityanalysispid}. As shown, the values of the string stability criteria increase significantly down the string of vehicles, indicating unstable behavior with this controller. The average time-domain string stability criteria using the \(US06\) drive cycle are presented in Table \ref{cacc.table2}.

\begin{table}[h]
\centering
\caption{Comparison of average time-domain string stability criteria under the \(US06\) drive cycle simulation.}\label{cacc.table2}
\begin{tabular}{cclc}
\multicolumn{2}{c}{\textbf{Controller}}&  & \textbf{\begin{tabular}[c]{@{}c@{}}Average time-domain\\string stability criteria\end{tabular}} \\ \hline
\multirow{3}{*}{\textbf{Proposed controller}} & \begin{tabular}[c]{@{}c@{}}\(b=0.1\) \end{tabular} &  & 1.0005                                  \\ \cline{2-4} 
  & \begin{tabular}[c]{@{}c@{}}\(b=0.2\)\end{tabular} &  & 1.0003     \\ \cline{2-4} 
  & \begin{tabular}[c]{@{}c@{}}\(b=0.5\)\end{tabular} &  & 0.9999            \\ \hline
\multirow{1}{*}{\textbf{Baseline controller}} & \begin{tabular}[c]{@{}c@{}}\(b=0.5\) \end{tabular} &  & 1.3638
\\ \hline
\end{tabular}
\end{table}

The proposed CACC guarantee safety and time-domain string stability at a headway of $b=0.5$~s, which is smaller than the $0.6-1.5$~s minimum reported for stable CACC platoons, enabling higher capacity platoons without compromising platoon stability and energy efficiency \cite{naus2010cacc,liu2020modifiedcacc,zhang2023delayreview,review_cth_policy}.
\subsection{\textbf{Energy Efficiency and String Stability Evaluation (Scenario II)}}
 
In this test scenario, the acceleration input for the first vehicle (leader) is shown in Figure \ref{fig:cacc_eff_sc}.
 A baseline PID-based CACC was enhanced to compare the effects of the designed controller on string stability and energy efficiency. This scenario was chosen to measure the impact of different controllers on string stability and energy efficiency. Therefore, string stability criteria were investigated through simulation, and experiments assessed energy efficiency.
    
  \begin{figure}[h]
    \centering
    \includegraphics[width=.36\textwidth]{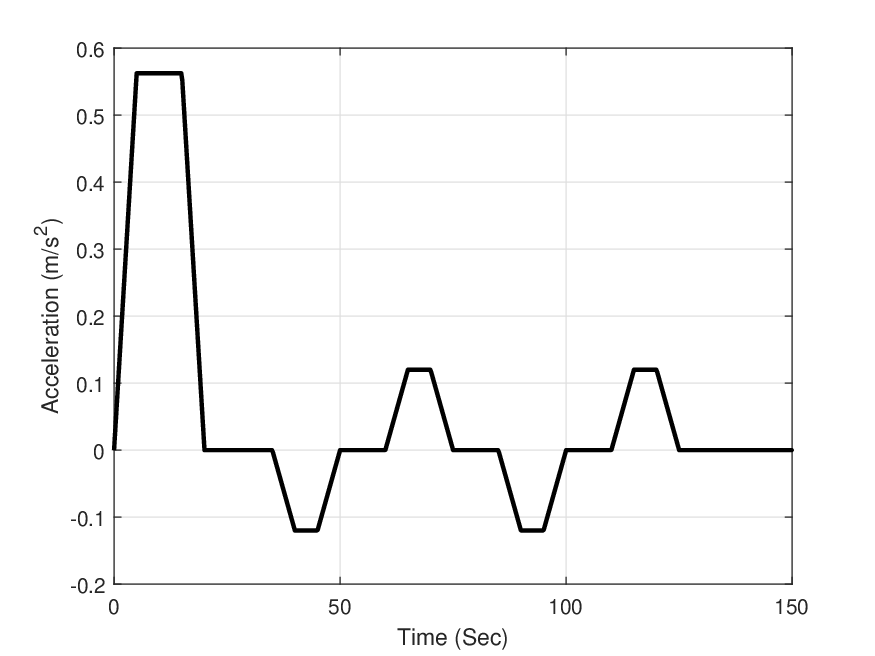}
    \caption{Desired acceleration input for the lead vehicle in scenario II.}
    \label{fig:cacc_eff_sc}
\end{figure}

\subsubsection{Simulation results for scenario II}
\

 The speed profiles of five vehicles in the simulation with the baseline controller and our proposed controller with \(b=0.5\) are shown in Figure \ref{fig:3_our_pid_sim}.

\begin{figure}[h] 
     \centering
    \begin{subfigure}{.36\textwidth}
        \includegraphics[width=\linewidth, right]{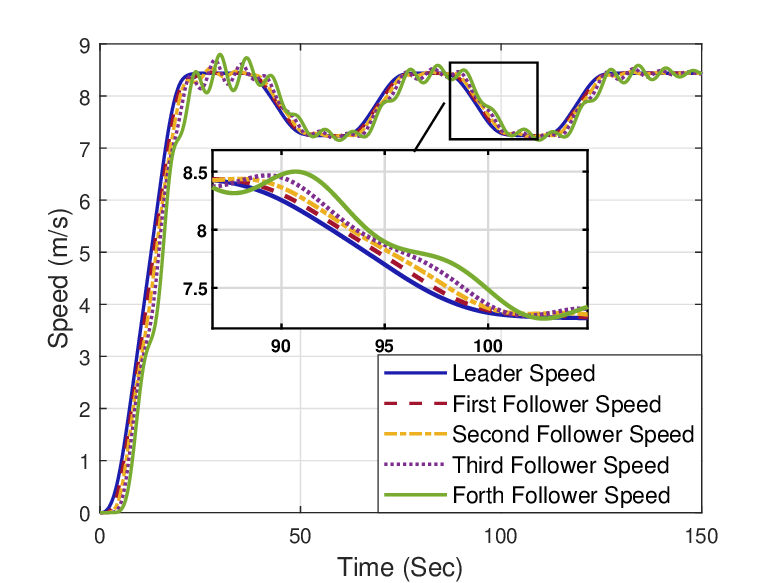}              \caption{}
        \
        \
    \end{subfigure}
    \begin{subfigure}{.36\textwidth}
            \includegraphics[width=\linewidth, left]{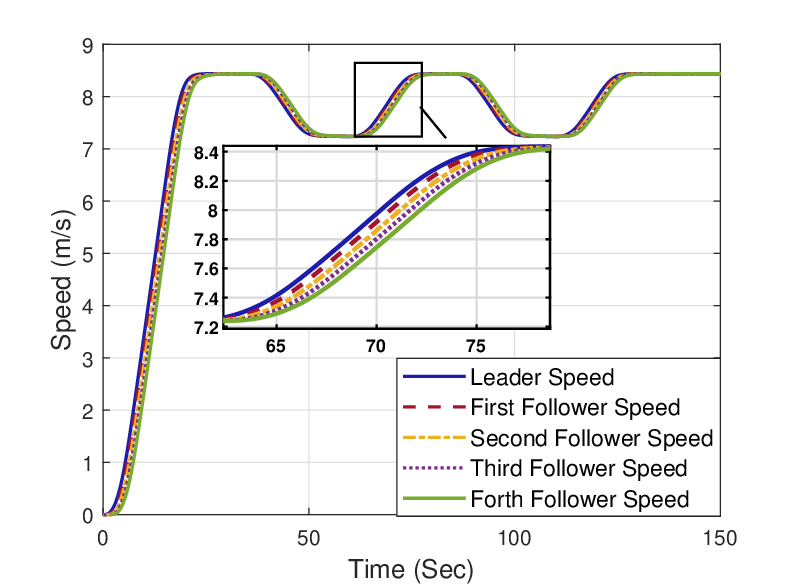}\caption{}
    \end{subfigure}
     \caption{Speed profiles of five vehicles in simulation in scenario II with 
    (a) the baseline controller,
    (b) the proposed controller.
       }
        
        \label{fig:3_our_pid_sim}
            \vspace{-15pt}
\end{figure}

To investigate this more, the string stability measure was studied using this scenario with these two controllers, this has been shown in Figure \ref{fig:3_ss}.
\begin{figure}[h]
    \centering
    \includegraphics[width=.37\textwidth]{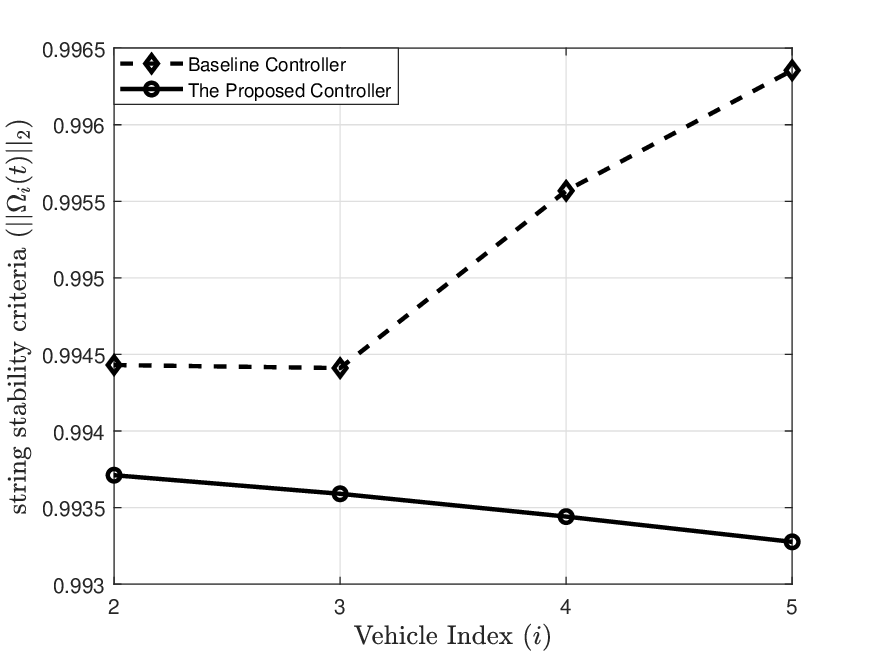}
    \caption{Comparison of time-domain string stability criteria for scenario II.}
    \label{fig:3_ss}
    \vspace{-10pt}
\end{figure}
As shown, in our proposed controller, the string stability criteria decrease down the string of vehicles, whereas the baseline controller shows an increase in the string stability criteria.
\subsubsection{Experimental results for scenario II}
\

To assess the experimental effect of the proposed controller on energy efficiency and string stability, we considered a string of five vehicles. The leader and the first three vehicles were modeled in a software-in-the-loop (SIL) environment, while the last vehicle was modeled in a VIL environment. This setup is depicted in Figure \ref{fig:cacc_eff_sc_exp}. The comparison of velocity profiles of the string of vehicles under this scenario in the experimental setup is shown in Figure \ref{fig:3_our_pid_exp}.
 
  \begin{figure}[h]
    \centering
\includegraphics[width=.516\textwidth]{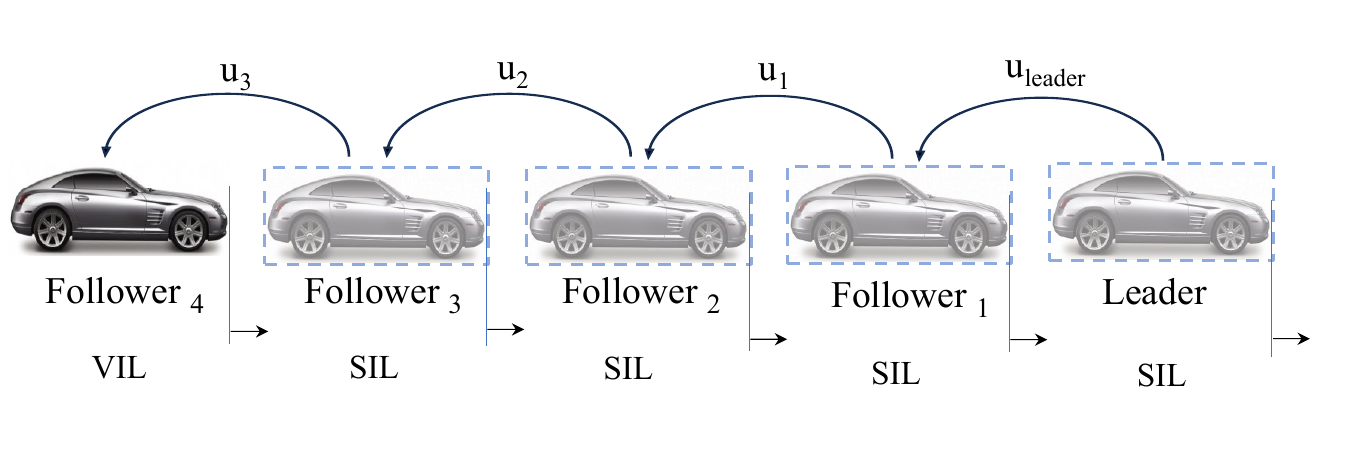}
    \caption{Configuration of vehicle string for experimental testing in scenario II}
    \label{fig:cacc_eff_sc_exp}
\end{figure}

\begin{figure}[h!] 
     \centering
    \begin{subfigure}{.35\textwidth}
        \includegraphics[width=\linewidth, right]{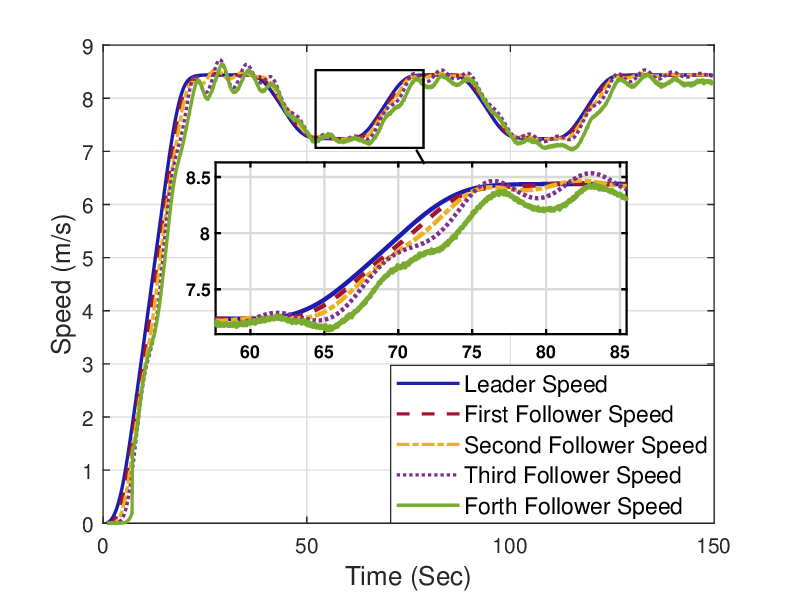}              \caption{}
        \
        \end{subfigure}
    \begin{subfigure}{.35\textwidth}
            \includegraphics[width=\linewidth, left]{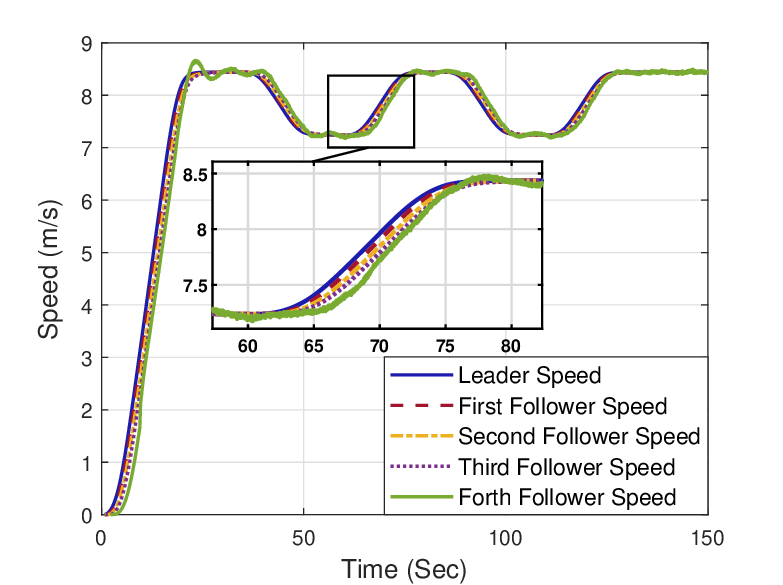}\caption{}
    \end{subfigure}
      \caption{Velocity profiles of five vehicles in experiment test scenario II with 
    (a) the baseline controller
    (b) the proposed controller
        }
        \label{fig:3_our_pid_exp}
            \vspace{-15pt}

\end{figure}

The Real-time power data, \(P(t)\), was obtained from vehicle sensors, measuring total power flow to the power management system and recorded in CAN messages. For test scenario II, the power consumption profile of both the baseline and proposed controllers is shown in Figure \ref{its.fig:experiment_eff_base}. In this plot, positive power values indicate energy consumption from the vehicle's energy management system, while negative values correspond to RBS, where energy is recaptured and fed back into the power management system. Since vehicle accessory systems, such as air conditioning, always consume a constant amount of power, we can assume that the difference in total net energy consumed between these two graphs illustrates how our proposed controller, which provides better string stability, enhances energy efficiency.
The calculated \( E_C \) values, derived from the data in Figure \ref{its.fig:experiment_eff_base} using \eqref{its.eq.power_integral}, are reported in Table \ref{cacc.table3} which reveals that improving string stability can reduce energy consumption by up to 38.5\%.

\begin{figure}[!h]
    \centering
    \includegraphics[width=.36\textwidth]{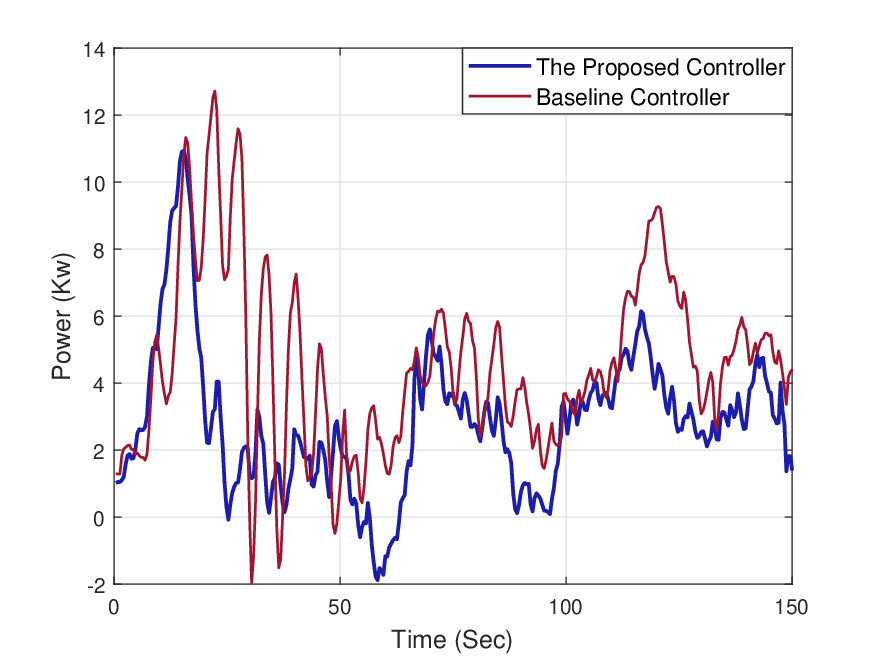}
    \caption{Comparison of average time-domain string stability criteria and total energy consumption between the proposed controller and the baseline controller, based on scenario II}
    \label{its.fig:experiment_eff_base}
\end{figure}

\begin{table}[h]\centering
\caption{Comparison of average time-domain string stability criteria and the average energy consumed in a platoon of 10 vehicles with the proposed controller and the baseline controller, based on the second scenario.}\label{cacc.table3}
\begin{tabular}{ccc}

\textbf{Controller}                                             & \textbf{\begin{tabular}[c]{@{}c@{}}Average time-domain\\string stability criteria  \end{tabular}} & \textbf{\begin{tabular}[c]{@{}c@{}}Average energy\\  consumed (KJ)  \end{tabular}} \\ \hline 
\begin{tabular}[c]{@{}c@{}}Proposed Controller\\\end{tabular} & 0.9921            & 421.07                          \\
\hline
Baseline controller               & 0.9979 
& 684.75
\\
\hline
\end{tabular}
\end{table}

  Platoon energy saving is a key metric. Prior CACC studies report about $4$--$10\%$ fuel reduction for conventional vehicles and $10$--$20\%$ energy savings for EV platoons by suppressing speed oscillations \cite{naus2010cacc,milanes2014cacc,zhao2020ev,liu2022evcacc}. In our results (Table~\ref{cacc.table3}), the proposed controller shows total energy 38.5\% improvement. For a representative $10$ EV platoon, this corresponds to the same system-level reduction, which is more than typical CACC gains reported in the literature.

\section{Conclusion and Future Works} \label{s:conclusion}
In CACC, the controller algorithm is crucial for ensuring string stability, energy efficiency, and robustness to fluctuations. However, existing models lacks accurately capture the distinct acceleration and regenerative braking dynamics of EVs, which are critical for precise control and energy optimization. To address this limitation, we proposed a new third-order EV model derived from real-world experimental data, explicitly distinguishing between motoring and regenerative braking phases. Based on this refined model, we developed a Lyapunov-based controller. The designed controller was shown to stabilize a string of at least five vehicles with a relatively small headway time in both simulations and experiments, thereby reducing the total energy consumption up to 38.5 \%. As the same CACC law governs all of vehicles, the per-vehicle energy reduction observed experimentally can be interpreted as the average saving per vehicle in a CACC platoon, affecting the entire platoon energy savings. It also highlighted that time-domain string stability analysis can show how the controller aids in energy efficiency, with lower average string stability criteria indicating lower energy consumption. This analysis can be based on standard drive cycles, like \(US06\), or any proposed cycle representing speed variations. Despite EVs being less impacted by speed fluctuations due to regenerative braking, such fluctuations can still lead to significant energy perturbations. Using the simulation method presented and evaluated in this paper, we have established a test and verification method for tuning headway time to ensure both string stability and energy efficiency.

One of the possible future works for this study could be investigating the effect of disturbances, like communication delay and noise in measuring data, and on necessary headway time for string stability and ensuring energy efficiency. Additionally, one can extend our analysis by deriving explicit frequency-domain conditions for string stability to gain deeper insights into the impact of control parameters on disturbance propagation in EV platoons.
\appendix \section*{Stability Analysis} \label{its.S:Stability}
\vspace{-5pt}
\subsection{Preliminary Calculations}

To prove stability, it is essential to first derive the relationships for the error dynamics, which will be utilized in the stability proof.
Taking the time derivative of the position error in \eqref{auxiliary_error} results
\begin{equation}\label{errors_derivative}
\dot{e}_{i_1}=r_{i_1}-\alpha_{i_1}e_{i_1},
\end{equation}
time derivative of auxiliary errors from \eqref{auxiliary_error_1} are obtained as
\begin{equation}\label{errors_derivatives}
\dot{r}_{i_1}=r_{i_2}-\alpha_{i_2}r_{i_1},
\end{equation}
and 
\begin{equation}\label{errors_derivatives_1}
\dot{r}_{i_2}=\ddot{r}_{i_1}+\alpha_{i_2}\dot{r}_{i_1},
\end{equation}
substituting \eqref{auxiliary_error} and \eqref{Errors} in \eqref{errors_derivatives_1} results 
\begin{equation}\label{errors_derivatives_2}
\dot{r}_{i_2}=(\dot{e}_{i_3}+\alpha_{i_1}e_{i_3})+\alpha_{i_2}(e_{i_3}+\alpha_{i_1}\dot{e}_{i_1}),
\end{equation}
and substituting \eqref{errors_derivatives}, yields
\begin{equation}\label{errors_derivatives_3}
\dot{r}_{i_2}=\dot{e}_{i_3}+\alpha_{i_1}e_{i_3}+\alpha_{i_2}e_{i_3}+\alpha_{i_2}\alpha_{i_1}r_{i_1}-\alpha_{i_2}\alpha_{i_1}^2e_{i_1},
\end{equation}
substituting $\dot{e}_{i_3}$ from \eqref{Acceleration_error_8} and more simplification make \eqref{errors_derivatives_3} as
\begin{equation}\begin{aligned}\label{errors_derivatives_4}
\dot{r}_{i_2}= & -\phi_i+ e_{i_3}(\alpha_{i_1}+\alpha_{i_2})-P_i+\beta_{i-1}u_{i-1} \\& + \alpha_{i_2}\alpha_{i_1}r_{i_1}-\alpha_{i_2}\alpha_{i_1}^2e_{i_1},
\end{aligned}
\end{equation}
substituting $P_i$ from \eqref{control_signal} results
\begin{equation}\label{errors_derivatives_5}
\dot{r}_{i_2}= -C_i\beta_i{r}_{i_2}-r_{i_1}.
\end{equation}
\vspace{-15pt}
\subsection{Notion of Solution, and Well-Posedness}
The longitudinal EV dynamic in \eqref{system_dynamics} is piecewise smooth and depend on the sign of the control input $u_i$ through
the coefficients $\gamma_i(u_i)$ and $\beta_i(u_i)$. As $\textit{sgn}(u_i)$ introduces a discontinuity on the switching surface $\{u_i = 0\}$, so the closed-loop system is a switched (nonsmooth) system. On each region where $u_i(t) \neq 0$, the  right hand side (RHS) of
\eqref{system_dynamics} is smooth and locally
Lipschitz in the state. At $u_i = 0$, we interpret the dynamics in the sense of
Filippov differential inclusions, i.e., $ \dot y \in F(y,t),$ where $F$ denotes the Filippov set-valued map associated with the RHS of \eqref{system_dynamics}.
On domain where $u_i \neq 0$, Filippov solutions coincide with classical Carath\'eodory solutions~\cite{Liberzon2003}. While, on the switching surface $u_i = 0$, the discontinuity in $\textit{sgn}(u_i)$ is considered as the convexification described in Filippov's construction~\cite{Filippov1988}, leading to absolutely continuous trajectory that may slide along the surface if needed. The generalized time derivative $\dot{\tilde V}_L$ used in the stability analysis is therefore understood in the Filippov sense. Using set-valued operator $K[\cdot]$ for non-smooth functions. This is precisely the framework adopted for the proof of stability, where
$K[\textit{sgn}(u_i)] = \mathrm{SGN}(u_i) \subseteq [-1,1]$ is used to bound $\dot{\tilde V}_L$ on and near the switching surface.

\begin{lemma}[Absence of Zeno switching]
\label{lem:non_zeno}
Consider the closed-loop dynamics of vehicle $i$ given by \eqref{system_dynamics}, with the control input $u_i(t)$ satisfying Remark~\ref{ass:regular_input}. Then, on any finite
time interval $[t_0,t_f]$, the switching signal
$\sigma_i(t) := \textit{sgn}(u_i(t)) \in \{-1,0,+1\}$ can switch only finitely many times. In particular, the closed-loop system admits no Zeno executions.
\end{lemma}

\begin{proof}
Fix a finite interval $[t_0,t_f]$ and omit the subscript $i$. Consider two consecutive zero crossings $t_k < t_{k+1}$ of $u(t)$ such that $\textit{sgn}(u)$ changes over $(t_k,t_{k+1})$. By Remark~\ref{ass:regular_input}(ii), there exists $t^\star \in (t_k,t_{k+1})$ with $|u(t^\star)| \geq \delta$. Without loss of
generality, assume $u(t^\star) \geq \delta$; the case $u(t^\star) \leq -\delta$
is analogous. Using the mean value theorem and the derivative bound $|\dot u(t)| \leq L$, we
obtain
\[
    t^\star - t_k
    \;\geq\; \frac{|u(t^\star) - u(t_k)|}{L}
    \;=\; \frac{|u(t^\star)|}{L}
    \;\geq\; \frac{\delta}{L},
\]
since $u(t_k)=0$. Similarly, because $u(t_{k+1}) = 0$, we have
\[
    t_{k+1} - t^\star
    \;\geq\; \frac{|u(t^\star) - u(t_{k+1})|}{L}
    \;\geq\; \frac{\delta}{L}.
\]
Therefore,
\begin{equation}
    t_{k+1} - t_k \;\geq\; \frac{2\delta}{L} \;=\;:~~ \kappa_{\min} > 0,
    \label{eq:min_dwell_time}
\end{equation}
where, $\kappa_{\min} \in \mathbb{R}$ is minium dwell time. Equation~\eqref{eq:min_dwell_time} shows that each change in $\textit{sgn}(u(t))$ requires at least $\kappa_{\min}$ units of time. Hence, on any finite interval $[t_0,t_f]$ the number of possible sign changes is upper bounded by $(t_f - t_0)/\kappa_{\min}$ and is therefore finite. Consequently, the switching signal $\sigma_i(t)$ cannot exhibit infinitely many switches in finite time, and ``Zeno'' executions are excluded.
\end{proof}

\subsection{Proof of Stability} \label{its.S:Stability_2}
Let the following be the sufficient conditions for the stability
\begin{equation}\begin{aligned}\label{sufficient-conditions}
&
\alpha_{i_2}>\frac{\varepsilon_{1_i}}{2},~~
\alpha_{i_1}>\frac{1}{2\varepsilon_{1_i}},~~
C_i>0,
\end{aligned}\end{equation}
where $\varepsilon_{1_i}\in\mathbb{R}^+$ denotes a positive known constant. 

\setcounter{theorem}{0}
\begin{theorem} \label{ITS.Th:main}
For the dynamic model described in \eqref{system_dynamics}, the controller developed in \eqref{control_signal} and \eqref{its.controller.P}, ensures asymptotically tracking, which means that the tracking errors in the Lyapunov functions converge to zero, given that the control gains satisfy the conditions in \eqref{sufficient-conditions}.
\end{theorem}
\begin{proof}

Let $y\triangleq[e_{i_1} ~~r_{i_1} ~~r_{i_2} ]^T \in \mathbb{D} \subseteq  \mathbb{R}^3$, and \( V_{L} : \mathbb{D} \times [0, +\infty) \to \mathbb{R} \) be a Lipschitz continuous, regular, positive definite function defined as
\begin{equation}\label{its.eq.proof1.1}
    V_L=\frac{1}{2}e_{i_1}^2+\frac{1}{2}r_{i_1}^2+\frac{1}{2}r_{i_2}^2,
\end{equation}
 which can satisfies the following inequalities
\begin{equation}\label{its.eq.proof1.3}
    U_{1}(y) \leq V_{L}(y) \leq U_{2}(y),
\end{equation}
where \( U_1(y) \) and \( U_2(y) \) are continuous positive definite functions, defined as
\begin{equation*}
    \begin{aligned}
    U_{1}(y) &= \frac{1}{2} \|y\|^2, \\
    U_{2}(y) &=  \|y\|^2.
\end{aligned}
 \end{equation*}

Since the dynamic model in \eqref{system_dynamics} is discontinuous in the set \( \{(y,t) \mid u = 0\} \), the existence and
stability of solutions cannot be analyzed using classical methods. To address this, we employ the Filippov differential inclusion, $\dot{y} \in F(y,t)$ where \( y \) is absolutely continuous, and \( F(\cdot) \) is Lebesgue measurable and locally bounded. Under Filippov's framework, a generalized Lyapunov stability theory can be applied to establish the robust stability \cite{filippov2013differential}. For almost all \( t \in [t_0, t_f] \), the generalized time derivative of \eqref{its.eq.proof1.1} exists almost everywhere (a.e.),
which can be written as
\begin{equation}\label{its.eq.proof1.100}
\dot{V}_L(y) \in \text{a.e. } \dot{\tilde{V}}_L(y),
\end{equation}
 ensuring that stability can be rigorously analyzed even in the presence of discontinuities in the system dynamics. Therefore,
\begin{equation} \label{its.eq.proof1.4}
    \dot{\tilde{V}}_L(y) = 
    \bigcap_{\chi \in \partial V_{L}(y)} 
    \chi^T K 
   \big[\dot{e}_{i_1}~~ \dot{r}_{i_1}~ ~\dot{r}_{i_2} \big]^T
\end{equation}
where
\begin{equation*}
    K[J](s, t) \triangleq \bigcap_{\delta > 0} \bigcap_{R(N) = 0} \overline{\text{co}} ~J(B(s, r)-N, t),
\end{equation*}
where $ \bigcap_{R(N) = 0} $ represents the intersection of all sets \( N \) of Lebesgue measure zero, \(\overline{\text{co}}\) denotes the convex closure, and \( B(s,r) \) is a ball of radius \( r \) centered at \( s \). Considering the fact that \( V_{L}(y) \) is a Lipschitz continuous regular function the operator is applicable. Therefore, we can simplify the RHS of \eqref{its.eq.proof1.4} as 
\begin{equation}\label{its.Stability.1}
\dot{\tilde{V}}_L(y)=  \big[e_{i_1}~~r_{i_1}~~r_{i_2} \big]K\big[\dot{e}_{i_1}~~\dot{r}_{i_1}~~\dot{r}_{i_2} \big]^T
\end{equation}
Utilizing the calculus framework for $K[.]$ as presented in \cite{paden1987calculus}, we can express \eqref{its.Stability.1} as follows
\begin{equation}\label{its.Stability_1}
\dot{\tilde{V}}_L(y) \subset e_{i_1}\dot{e}_{i_1}+r_{i_1}\dot{r}_{i_1}+r_{i_2}\dot{r}_{i_2}.
\end{equation}
using time derivatives of all errors in \eqref{errors_derivative}, \eqref{errors_derivatives}, and \eqref{errors_derivatives_5}
we can write \eqref{its.Stability_1} as
\begin{equation}\begin{aligned}\label{Stability_2}
\dot{\tilde{V}}_L(y) \subset &~ e_{i_1}(r_{i_1}-\alpha_{i_1}e_{i_1})+r_{i_1}(r_{i_2}-\alpha_{i_2}r_{i_1}) \\& ~~~+ r_{i_2}(-C_i\beta_i{r}_{i_2}-r_{i_1}),
\end{aligned}
\end{equation}
simplification of \eqref{Stability_2} gives
\begin{equation}\label{its.Stability_3}
-\alpha_{i_2}r_{i_1}^2-C_i\beta_i{r}_{i_2}^2-\alpha_{i_1}e_{i_1}^2+e_{i_1}r_{i_1}.
\end{equation}

By using \eqref{system_dynamics}, the expression in 
\eqref{its.Stability_3} can be written as
\begin{equation}\label{its.Stability_3.1}
\begin{aligned}
\dot{\tilde{V}}_L(y) \subset~&-\alpha_{i_2}r_{i_1}^2-\alpha_{i_1}e_{i_1}^2\\&-C_i(\frac{1 + K[\textit{sgn}(u_i)]}{2} \beta_{i_1} + \frac{1 - K[\textit{sgn}(u_i)]}{2} \beta_{i_2}){r}_{i_2}^2
\\&+e_{i_1}r_{i_1},
\end{aligned}
\end{equation}
where $K[\textit{sgn}(u_i)]=\text{SGN}(u_i)$ can be calculated as \cite{paden1987calculus}
\begin{equation*}
     \text{SGN}(u_i)=\begin{cases}
         1, ~~~~~~~~~~\text{if}~~ u_i >0,\\
         [-1,1] ~~~~~\text{if} ~~u_i =0,\\
         -1, ~~~~~~~~\text{if}~~ u_i <0.
     \end{cases}
\end{equation*}
Simplification of \eqref{its.Stability_3.1} yields
\begin{equation}\label{its.Stability_3.2}
\begin{aligned}
\dot{\tilde{V}}_L(y) \subset~&-\alpha_{i_2}r_{i_1}^2-\alpha_{i_1}e_{i_1}^2
+e_{i_1}r_{i_1}
\\&-\frac{C_i}{2}\bigg((\beta_{i_1} +\beta_{i_2})+ (\beta_{i_1} -\beta_{i_2}) \text{SGN}(u_i)\bigg){r}_{i_2}^2.
\end{aligned}
\end{equation}

Defining $\delta_i\triangleq\bigg((\beta_{i_1} +\beta_{i_2})+ (\beta_{i_1} -\beta_{i_2}) \text{SGN}(u_i)\bigg)$, we can write \eqref{its.Stability_3.2} as
\begin{equation}\label{its.Stability_3.2.2.3}
\begin{aligned}
\dot{\tilde{V}}_L(y) \subset&~-\alpha_{i_2}r_{i_1}^2-\alpha_{i_1}e_{i_1}^2
+e_{i_1}r_{i_1}
-\frac{\delta_i C_i}{2}{r}_{i_2}^2.
\end{aligned}
\end{equation}

\begin{remark}
\label{its.rem.1} As mentioned in Section \ref{cacc_model}, $\beta_{i_1}, \beta_{i_2} \geq 0 $, therefore we can write 
\begin{equation}\label{its.Stability_3.2.1}
\begin{aligned}
\delta_i\geq \big((\beta_{i_1} +\beta_{i_2})- |\beta_{i_1} -\beta_{i_2}|\big) > 0.
\end{aligned}
\end{equation}
\end{remark} 

By defining $m_i\triangleq \frac{1}{2}((\beta_{i_1} +\beta_{i_2})- |\beta_{i_1} -\beta_{i_2}|)$ and noting from Remark \ref{its.rem.1} that \( m_i \) is a positive constant, we can then apply \eqref{its.Stability_3.2.2.3} to obtain
\begin{equation}\begin{aligned}\label{Stability_4}
\dot{\tilde{V}}_L \leq & -\alpha_{i_2}\norm{r_{i_1}}^2-\alpha_{i_1}\norm{e_{i_1}}^2 +
\norm{e_{i_1}}\norm{r_{i_1}}\\&-C_im_i\norm{r_{i_2}}^2.
\end{aligned}
\end{equation}
Young's Inequality is applied to select terms in \eqref{Stability_4} as
\begin{equation}\label{Stability_5}
\norm{e_{i_1}}\norm{r_{i_1}} \leq \frac{1}{2\varepsilon_{1_i}}\norm{e_{i_1}}^2+\frac{\varepsilon_{1_i}}{2}\norm{r_{i_1}}^2, 
\end{equation}
using \eqref{Stability_5}, the inequality in \eqref{Stability_4} can be written as 
\begin{equation}\begin{aligned}\label{Stability_6}
\dot{\tilde{V}}_L(y)\leq & -\alpha_{i_2}\norm{r_{i_1}}^2-C_im_i\norm{r_{i_2}}^2-\alpha_{i_1}\norm{e_{i_1}}^2 \\&
+ \frac{1}{2\varepsilon_{1_i}}\norm{e_{i_1}}^2+\frac{\varepsilon_{1_i}}{2}\norm{r_{i_1}}^2.
\end{aligned}
\end{equation}

Simplifying \eqref{Stability_6} results 
\begin{equation}\begin{aligned}\label{Stability_7}
\dot{\tilde{V}}_L(y)\leq & -(\alpha_{i_2}-\frac{\varepsilon_{1_i}}{2})\norm{r_{i_1}}^2 -(\alpha_{i_1}-\frac{1}{2\varepsilon_{1_i}})\norm{e_{i_1}}^2 \\&
-C_im_i\norm{r_{i_2}}^2.
\end{aligned}
\end{equation} 
Satisfying the sufficient conditions in \eqref{sufficient-conditions} and Remark \ref{its.rem.1}, yields finding $U(y)\triangleq\lambda|| y||^2$ such that
\begin{equation}\begin{aligned}\label{Stability_8}
\dot{\tilde{V}}_L(y)\leq  -U(y)=-\lambda\norm{y}^2,\\
\end{aligned}
\end{equation}
where $\lambda=\text{min}\{\alpha_{i_2}-\frac{\varepsilon_{1_i}}{2}, \alpha_{i_1}-\frac{1}{2\varepsilon_{1_i}},  C_im_i\}$. The inequality in  \eqref{Stability_8} will lead to
\begin{equation}\begin{aligned}\label{its.Stability_9}
\dot{{V}}_L(y)\leq-U(y),~~~~~~~\forall \dot{V}(y) \mathrel{\overset{\text{a.e.}}{\in}} \dot{\tilde{V}}(y)
\end{aligned}
\end{equation}
which results in
\begin{equation}\begin{aligned}\label{its.Stability_10}
\dot{{V}}_L(y)\leq-\lambda \norm{y}^2,
\end{aligned}
\end{equation}
by considering bonds in \eqref{its.eq.proof1.3}, we can write \eqref{its.Stability_10} as
\begin{equation}\begin{aligned}\label{its.Stability_11}
\dot{{V}}_L(y)\leq-\lambda V_L
\end{aligned}
\end{equation}

By designing the controller is derived from \eqref{control_signal} and considering \eqref{its.Stability_11}, which means $\dot{V}$ is made negative definite, ensuring $e_{i_1}, r_{i_1},$ and $r_{i_2}$ converges to zero, thus achieving asymptotic tracking.

\end{proof}

\begin{remark}[Filippov solutions and stability of the switched dynamics]
\label{rem:filippov_stability} Under Remark~\ref{ass:regular_input}, the closed-loop EV dynamics define a Filippov differential inclusion as in  with a locally bounded, measurable right-hand side. Together with the non-Zeno
property established in Lemma~\ref{lem:non_zeno}, this guarantees the existence
of complete, absolutely continuous Filippov solutions for all initial
conditions of interest.

The generalized derivative $\dot{\tilde V}_L$, is computed with respect
to this Filippov inclusion, using $K[\textit{sgn}(u_i)] = \mathrm{SGN}(u_i)$ on the
switching surface. Therefore, Theorem~\ref{ITS.Th:main} establishes
asymptotic tracking and time-domain string stability for all Filippov solutions
of the switched EV dynamics, while ensuring that no Zeno or chattering
behavior arises from the regenerative braking switching mechanism.
\end{remark}

\section*{acknowledgment}

Partial support for this research was provided by the National Science Foundation under Grant No. ECCS-EPCN-2241718. Any opinions, findings, conclusions, or recommendations expressed in this material are those of the author(s) and do not necessarily reflect the views of the sponsoring agency.

\bibliography{Ref.bib}
\bibliographystyle{ieeetr}
\end{document}